\documentclass[conference,a4paper]{IEEEtran} 
\usepackage{cite}
\usepackage[dvips]{graphicx}
\usepackage{cite}
\usepackage{amsmath}
\usepackage{amsfonts}
\usepackage{amssymb}
\usepackage{color}
\usepackage{amsthm}

\newtheorem*{lemma*}{Lemma}
\newtheorem*{exmp*}{Example}

\newcommand{\titletext}{On Continuous-Time White Phase Noise Channels}

\newcommand*{\diff}{\mathop{}\mathopen{\mathrm{d}}}

\newcommand{\Exp}{\mathsf E}
\newcommand{\expect}[1]{{\Exp}\left[#1\right]}

\newcommand{\Var}{\mathsf{Var}}
\newcommand{\variance}[1]{{\Var}\left[#1\right]}

\begin{document}

\sloppy
%
\title{\titletext}
%

\author{
\authorblockN{Luca Barletta}
\authorblockA{Institute for Advanced Study\\
Technische Universit\"{a}t M\"{u}nchen \\
D-85748 Garching, Germany\\
luca.barletta@tum.de} \and
\authorblockN{Gerhard Kramer}
\authorblockA{Institute for Communications Engineering\\
Technische Universit\"{a}t M\"{u}nchen\\
D-80333 Munich, Germany\\
gerhard.kramer@tum.de}}

%
%
%
%
%
%
\maketitle
\begin{abstract}
A continuous-time model for the additive white Gaussian noise (AWGN) channel in the presence
of white (memoryless) phase noise is proposed and discussed.
It is shown that for linear modulation the output of the baud-sampled filter matched to the shaping waveform
represents a sufficient statistic. The analysis shows that the phase noise channel has the same information rate
as an AWGN channel but with a penalty on the average signal-to-noise ratio, the amount of penalty depending on
the phase noise statistic.
\end{abstract}
%
%
%
%
\section{Introduction}

A major source of impairment in
radio and optical channels is multiplicative phase noise that arises due to instabilities of local
oscillators used in up- and down-conversion.
Estimation and compensation of phase noise is a classic problem in communication systems theory.
In fiber-optic communication systems, phase noise is present due to 
instabilities of laser oscillators~\cite{Foschini88} or due to cross-phase
modulation in multi-channel systems~\cite{Essiambre2010}.

In the literature the phase noise process is often modeled as a discrete-time process that
is output by a baud-sampled matched filter~\cite{Spalvieri2011, Barletta2012}. This 
approximate model is valid as long as the variations of the phase noise process
are small in one symbol period and the signal-to-noise ratio (SNR) is low or moderate~\cite{Ghozlan2013}. In general, and especially for phase
noise characterized by large linewidths, the channel impaired by phase noise should
be modeled and studied in continuous-time, in order to account for the effects
of filtering prior to sampling~\cite{GhozlanISIT2013, Ghozlan2013, GhozlanISIT2014}.

Analytical bounds on the capacity of phase noise channels are given in~\cite{Katz2004}, and for high SNR in~\cite{Lapidoth2002}.
Numerical bounds on the information rate transferred through discrete-time Wiener phase noise channels
are provided in~\cite{Barletta2012, Barletta2011, Barbieri2011}, while for the continuous-time counterpart analytical lower bounds are given
in~\cite{GhozlanISIT2013, GhozlanISIT2014} and an analytical upper bound in~\cite{CrownCom2014}. Numerical results for the achievable information rates in the case
of discrete-time phase noise with an arbitrary spectrum are given in~\cite{Barletta2013, BarlettaTIT2014}.

Few papers in the literature deal with continuous-time phase noise and, apart from
trivial cases, the existence of a finite-dimensional sufficient statistic for these channels 
has not been investigated yet. The aim of this paper is to attempt to model the continuous-time additive white Gaussian noise (AWGN) channel
in the presence of \emph{white} phase noise, and to find a (finite-dimensional) sufficient statistic. 

The motivation for studying white phase noise is that there are problems for which the phase varies more rapidly than the bandwidth of the receiving device, 
e.g., a ``square law'' detector that necessarily has limited bandwidth. In such cases, the device will ``lose'' energy and we wish to model this effect.

The paper is organized as follows. In Sec.~\ref{Sec:previous} previous works on this topic are analyzed in detail.
The system model is introduced and discussed in Sec.~\ref{Sec:system}, while a sufficient statistic
is derived in Sec.~\ref{Sec:sufficient}. Results are discussed in Sec.~\ref{Sec:discussion}, and conclusions are drawn in Sec.\ref{Sec:conclusion}.

\section{Previous Work}\label{Sec:previous}
The output of a baud-sampled matched filter represents a sufficient statistic for 
the continuous-time AWGN channel with a constant phase shift given by a random parameter (see, e.g., \cite[Sec. 4.4.1]{VanTrees68}).

In the context of continuous-time Wiener phase noise channels, the authors of~\cite{Ghozlan2013} show through numerical simulations that 
a multi-sample receiver achieves significantly higher information rates than the baud-sampled matched filter, especially at high SNR.
The dimensionality of the sufficient statistic for this channel with memory has not been investigated.

In the context of continuous-time memoryless phase noise channels, the authors of~\cite[Eq. (28)]{Goebel2011} model the \emph{partially coherent} channel as
\begin{equation}
 Y(t) = X(t) e^{j\Theta(t)} + W(t), \qquad 0\le t\le T
\end{equation}
where $W$ is a complex-valued AWGN process and $\Theta$ is a white phase noise process used to model 
the nonlinear effects of cross-phase modulation in multichannel fiber-optic communication systems.
\emph{White} means that $\Theta(t_1)$ and $\Theta(t_2)$ are uncorrelated random variables 
for any $t_1, t_2$ in the interval $[0,T]$ with $t_1\ne t_2$. If, in addition, $\Theta$ is stationary and
$\Theta(t)$ is distributed
according to a wrapped Gaussian with zero mean and variance $\sigma^2$, then the autocorrelation function $R_{e^{j\Theta}}$ of
the multiplicative disturbance $e^{j\Theta(t)}$ is
\begin{align}
 R_{e^{j\Theta}}&=\expect{e^{j\Theta(t)}e^{-j\Theta(t+\tau)}}=\left\{\begin{array}{ll}
                                                     1 & \tau=0 \\
                                                     e^{-\sigma^2} & \tau\ne 0
                                                    \end{array}\right.\\
                &=e^{-\sigma^2}+\lim_{B\rightarrow\infty} (1-e^{-\sigma^2})\cdot \text{sinc}(B\tau).
\end{align}
The power spectral density (PSD) $S_y(f)$ of the received signal $Y(t)$ is given by
\begin{align}
 S_y(f) &={\cal F}\{R_y(\tau)\}= S_x(f)\star S_{e^{j\Theta}}(f)+S_w(f)\nonumber\\
 &= e^{-\sigma^2} S_x(f)+\lim_{B\rightarrow\infty} S_x(f)\star (1-e^{-\sigma^2})\frac{\text{rect}(f/B)}{B}\nonumber\\
 &\quad+S_w(f), \label{eq:psd}
\end{align}
where ${\cal F}\{\cdot\}$ denotes the Fourier transform and $\star$ the convolution operator. Equation~(\ref{eq:psd}) shows that
the PSD of the input signal $S_x(f)$ experiences a \emph{spectral loss}~\cite{Goebel2011}, that is a gain equal to $e^{-\sigma^2}\le 1$, and
the remaining power proportional to ($1-e^{-\sigma^2})$ is spread onto an infinite bandwidth. Later in this paper we will see that when using projection
receivers the spectral loss is unavoidable.

\section{System Model}\label{Sec:system}
Consider linear modulation for which we express the output of the continuous-time white phase noise channel as
\begin{align}
 Y(t) &= X(t) e^{j\Theta(t)} + W(t) \nonumber\\
  &= \sum_k A_k g(t-kT) e^{j\Theta(t)} + W(t) \nonumber\\
 &=\sum_k A_k g_k(t) e^{j\Theta(t)} + W(t),\label{eq:model}
\end{align}
where $\{A_k\}$ is a sequence of independent and identically distributed (iid) random variables, $T$ is the symbol period, $j$ is the imaginary unit, $g(t)$ is the pulse shaping filter 
such that $\{g_k(t)\}$ is a 
set of orthonormal functions of $L^2(\mathbb{R})$ with the property
\begin{equation}\label{eq:translation}
 \int_{-\infty}^\infty g_k(t)g_l^\star(t)\diff t = \left\{\begin{array}{ll}
                                                     1 & k=l \\
                                                     0 & k\ne l
                                                    \end{array}\right.
\end{equation}
where the superscript $^\star$ denotes complex conjugate,
and $W$ is a complex-valued circularly symmetric white Gaussian noise.
Let 
\begin{equation}\label{eq:basis}
  \{\phi_{nm}(t), n=0,1,\ldots; m\in\mathbb{Z}\}
\end{equation}
be a complete orthonormal basis of $L^2(\mathbb{R})$ with \mbox{$\phi_{nm}(t)=\phi_n(t-mT)$} and such that
\begin{equation}\label{eq:double}
 \int_{-\infty}^\infty \phi_{nm}(t)\phi_{n'm'}^\star(t)\diff t = \left\{\begin{array}{ll}
                                                     1 & n=n',m=m' \\
                                                     0 & \text{otherwise},
                                                    \end{array}\right.
\end{equation}
and
\begin{equation}\label{eq:proj_g_phi}
 \int_{-\infty}^\infty g_k(t)\phi_{nm}^\star(t)\diff t = \left\{\begin{array}{ll}
                                                     1 & n=0,m=k \\
                                                     0 & \text{otherwise}.
                                                    \end{array}\right.
\end{equation}
The basis functions in~(\ref{eq:basis}) are double indexed: the index $nm$ denotes the $n-$th basis function $\phi_n(t)$ translated by $mT$ in time, i.e., $\phi_n(t-mT)$.
The choice of using a double index for basis functions is because orthogonality is also obtained by translation, as shown in~(\ref{eq:double}).
Property~(\ref{eq:proj_g_phi}) states that the pulse shaping function $g(t)$ has a non-zero projection only in the first dimension, i.e., for $n=0$. This means that \mbox{$\phi_{0m}(t)=g_m(t)$}. Note that
assuming property~(\ref{eq:proj_g_phi}) comes without loss of generality, because we can always obtain all the other basis functions $\phi_n(t)$ for $n>0$ with an orthonormalization procedure.
\begin{exmp*}
 A complete orthonormal basis of $L^2(\mathbb{R})$ is the trigonometric system
\begin{equation}
 \left\{\phi_{nm}(t) = \phi_n(t-mT),n=0,1,\cdots; m\in \mathbb{Z}\right\},
\end{equation}
where
\begin{equation}
 \phi_n(t) = \frac{1}{\sqrt{T}}\exp\left(j\frac{2\pi}{T}nt\right)
\end{equation}
for $0\le t< T$ and $\phi_n(t)$ is zero elsewhere. The square pulse shaping filter
\begin{equation}
 g_0(t) = \left\{\begin{array}{ll}
                                                     1/\sqrt{T} & 0\le t < T \\
                                                     0 & \text{otherwise}
                                                    \end{array}\right.
\end{equation}
is such that~(\ref{eq:translation}) and~(\ref{eq:proj_g_phi}) are satisfied.
\end{exmp*}

The phase noise process $\Theta$ is defined as a stationary random process such that
\begin{equation}\label{eq:average}
 \expect{\left\langle g_k(t)e^{j\Theta(t)},\phi_{nm}(t)\right\rangle} \triangleq \mu_\Theta \cdot \left\{\begin{array}{ll}
                                                     1 & n=0,m=k \\
                                                     0 & \text{otherwise},
                                                    \end{array}\right.
\end{equation}
\begin{align}\label{eq:variance}
 &\expect{\left\langle g_k(t)e^{j\Theta(t)},\phi_{nm}(t)\right\rangle\cdot \left\langle g_{k'}(s)e^{-j\Theta(s)},\phi_{n'm'}(s)\right\rangle}  \nonumber\\
 &\triangleq |\mu_\Theta|^2\cdot \left\{\begin{array}{ll}
                                                     1 & n=n'=0,m=k,m'=k', \\
                                                     0 & \text{otherwise},
                                                    \end{array}\right.
\end{align}
where $\mu_\Theta = \expect{e^{j\Theta(t)}}$.
Note that~(\ref{eq:average}) and~(\ref{eq:variance}) give 
\begin{equation}\label{eq:variance2}
 \variance{\left\langle g_k(t)e^{j\Theta(t)},\phi_{nm}(t)\right\rangle} = 0
\end{equation}
for any choice of $k$, $n$, and $m$. In other words, we can view the projections $\left\langle g_k(t) e^{j\Theta(t)}, \phi_{nm}(t) \right\rangle$ as constants.

Defining the phase noise process through~(\ref{eq:average})-(\ref{eq:variance2}) is reasonable
if we interpret the model (\ref{eq:model}) as the limit of successive refinements of a related discrete-time model. 
For the discrete-time model, consider the finite interval $[-S,S]$ and time instants $t_i=iS/l$ for $i=-l,\ldots,l-1$.
The extension to the whole real axis is straightforward by letting $S\rightarrow\infty$.
Since the continuous-time AWGN $W(t)$ is defined only when using projection receivers,
also for the phase noise process it is desirable to preserve the statistical properties of its projections onto
the basis functions, when passing from the discrete-time to the continuous-time model. Mathematically speaking, it is useful to define
\begin{equation}\label{eq:definition}
 \left\langle g_k(t)e^{j\Theta(t)},\phi_{nm}(t)\right\rangle \triangleq \lim_{l\rightarrow\infty}  \frac{1}{2l}\sum_{i=-l}^{l-1} g_k(t_i) e^{j\Theta(t_i)} \phi_{nm}^\star(t_i),
\end{equation}
where the equality represents some kind of convergence.
The condition and the type of convergence is given by the following Lemma.
\begin{lemma*}\label{lm:projection} If $|g_k(t)\phi_{nm}(t)|$ is uniformly bounded in $t$, then the right-hand side of (\ref{eq:definition}), that is the projection 
of the phase noise process $g_k(t)e^{j\Theta(t)}$ onto the basis function $\phi_{nm}(t)$ of $L^2([-S,S])$, converges (almost surely) to
the expected value of the projection:
\begin{equation}\label{eq:lemma}
  \left\langle g_k(t)e^{j\Theta(t)},\phi_{nm}(t)\right\rangle \stackrel{a.s.}{=} \mu_\Theta \int_{-S}^S g_k(t)\phi_{nm}^\star(t)\diff t,
\end{equation}
for any value of $k$, $n$, and $m$.
\end{lemma*}
\begin{proof} We show the convergence result for the real part of the projection; the procedure for the imaginary part is analogous.
 Consider the independent random variables 
 \begin{equation}
    Z_{i+1} = \Re\{g_k(t_i)\phi_{nm}^\star(t_i)e^{j\Theta(t_i)}\},\qquad i=0,\ldots,l-1,
 \end{equation}
 where $\Re\{x\}$ is the real part of complex number $x$,
 then Kolmogorov's criterion for the strong law of large numbers~\cite[Th. 1-3, pag. 238]{feller71} for $l\rightarrow\infty$ reads as
 \begin{align}
  \sum_{i=1}^\infty \frac{\variance{Z_i}}{i^2}&\le\sum_{i=1}^\infty \frac{\expect{\Re\{g_k(t_{i-1})\phi_{nm}^\star(t_{i-1}) e^{j\Theta(t_{i-1}) } \}^2} }{i^2} \nonumber\\ 
  &\le \sum_{i=1}^\infty \frac{|g_k(t_{i-1})\phi_{nm}(t_{i-1}) |^2 }{i^2} \nonumber\\
  &\le\frac{K\pi^2}{6}<\infty,\label{eq:app3}
 \end{align}
 where~(\ref{eq:app3}) follows by the boundedness condition $|g_k(t)\phi_{nm}(t)|^2\le K$ for any $t\in[-S,S]$,
 with $K<\infty$. The finiteness of~(\ref{eq:app3}) implies the strong law of large numbers for independent and non-identically distributed random variables:
 \begin{align}
  \frac{1}{l}\sum_{i=1}^{l} Z_i \stackrel{a.s.}{\longrightarrow} &\lim_{l\rightarrow\infty}\frac{1}{l}\sum_{i=1}^{l}\expect{Z_i}\nonumber\\
  &= \Re\left\{\mu_\Theta\int_0^S g_k(t)\phi_{nm}^\star(t)\diff t\right\}
 \end{align}
 where the last equality follows by the definitions of $\mu_\Theta$ and of the Riemann integral.
 
 Proving the convergence in the time interval $[-S,0]$ is analogous, and the two parts together prove the thesis.
\end{proof}
Result~(\ref{eq:lemma}) is in accordance 
with the definitions~(\ref{eq:average})-(\ref{eq:variance2}). The condition of boundedness of the product of the basis functions required by the Lemma is mild and can be
easily met, therefore we assume it is met. 

Note that the proof of the Lemma requires a strong condition: independence among phase noise samples $\Theta(t_i)$.
But this independence condition is not unrealistic, as phase noise samples can be uncorrelated and Gaussian distributed, and thus independent. 
This can be the case, e.g., with phase noise generated by cross-phase modulation caused by neighboring channels in 
multichannel fiber-optic communication systems~\cite{Essiambre2010, Goebel2011}.

\section{Sufficient Statistic}\label{Sec:sufficient}
The average mutual information, per unit time, between the input and the output of the continuous-time channel in~(\ref{eq:model})
is defined as~\cite[pag. 370]{Gallager68}
\begin{equation}
 I(X;Y) \triangleq \lim_{M\rightarrow\infty}\frac{1}{2M+1} \lim_{N\rightarrow\infty} I(X^{N,M};Y^{N,M}),
\end{equation}
where 
\begin{equation}
 X^{N,M} = (X_{0}^M,X_{1}^M,\ldots,X_N^{M}),
\end{equation}
and
\begin{equation}
X_n^M=(X_{n,-M},X_{n,-M+1},\ldots,X_{nM}), 
\end{equation}
where
\begin{align}
 X_{nm} &= \left\langle X(t), \phi_{nm}(t) \right\rangle\\
 &= \sum_k A_k \left\langle g_k(t), \phi_{nm}(t)\right\rangle= \left\{\begin{array}{ll}
                     A_m & n=0, \\
                     0 & n > 0.
                 \end{array}\right. \label{eq:proj_x}
\end{align}
In other words, $X^{N,M}$ is a vector that collects the \mbox{$(N+1)\cdot (2M+1)$} projections of process $X$
onto the basis functions $\phi_{nm}(t)$ for \mbox{$n=0,1,\ldots,N$} and \mbox{$m=-M,-M+1,\ldots,M$}. An analogous notation holds for process $Y$.

The projection of $Y$ onto $\phi_{nm}(t)$ is
\begin{align}
 Y_{nm} &= \left\langle Y(t), \phi_{nm}(t)\right\rangle \\
 &= \left\langle X(t)e^{j\Theta(t)}+W(t), \phi_{nm}(t)\right\rangle \\
 &= \sum_k A_k \left\langle g_k(t)e^{j\Theta(t)}, \phi_{nm}(t)\right\rangle + W_{nm}\\
 &\stackrel{a.s.}{=} \mu_\Theta \sum_k A_k \int_{-\infty}^\infty g_k(t)\phi_{nm}^\star(t)\diff t + W_{nm} \label{eq:proj_y}
\end{align}
where the $W_{nm}$'s are iid circularly-symmetric Gaussian variates, and
the last step follows by the Lemma in the previous section. By using~(\ref{eq:proj_g_phi}) in~(\ref{eq:proj_y}) we have
\begin{equation}\label{eq:proj_y2}
 Y_{nm} \stackrel{a.s.}{=} \left\{\begin{array}{ll}
                     \mu_\Theta A_m + W_{nm} & n=0, \\
                     W_{nm} & n> 0.
                 \end{array}\right.
\end{equation}

Since $Y_{nm}$ does not carry information about the input $A_m$ for $n>0$, and the sequences of random variables converge 
almost surely to their limits, we have
\begin{align}
 I(X;Y) &= \lim_{M\rightarrow\infty}\frac{1}{2M+1} I(X_0^M;Y_0^M) \\
 &= \lim_{M\rightarrow\infty}\frac{1}{2M+1} \sum_{m=-M}^M I(X_{0m};Y_{0m}) \label{eq:mutinf1}\\
 &=  I(X_{00};Y_{00}) = I(A_0; Y_{00})\label{eq:mutinf2}
\end{align}
where~(\ref{eq:mutinf1}) follows because the channel is memoryless by~(\ref{eq:proj_y2}), and~(\ref{eq:mutinf2}) follows
by stationarity and by~(\ref{eq:proj_x}). The result~(\ref{eq:mutinf2}) shows that a sufficient statistic for inferring $A_k$ given $Y$ is
\begin{equation}
 Y_{0k} = \left\langle Y(t), \phi_{0k}(t)\right\rangle = \left\langle Y(t), g_{k}(t)\right\rangle
\end{equation}
that is the output of the sampled filter that is matched to $g(t)$.

\section{Discussion}\label{Sec:discussion}
By virtue of~(\ref{eq:proj_y2}), at the output of the sampled matched filter the continuous-time white phase noise channel is equivalent to a
discrete-time AWGN channel
\begin{equation}
 Y_{0k} = \mu_\Theta A_k + W_{0k}
\end{equation}
so the effect of the spectral loss $\mu_\Theta$ is to reduce the SNR. In the limit of $\mu_\Theta\rightarrow 0$, reliable communication over this
channel is not possible. 

We conclude that the presence of rapidly varying continuous-time phase noise affects not only the information
encoded in the \emph{phase} of $X(t)$, but also the information encoded in its \emph{amplitude}. We further conclude that even a noncoherent 
receiver is affected by white phase noise and suffers from an SNR loss. This is true also for fiber-optic communication
systems based on \emph{direct detection}, e.g., when the front-end is a photodiode that performs photon counting, the numbers of photons being
proportional to the energy the photodiode harvests. The output of the photodiode is usually modeled as the modulus-square of the 
field at the input, so in principle the phase noise would not be an issue. But real photodiodes are bandlimited and should be modeled as being 
preceded by a bandpass filter that 
reduces the energy collected by the receiver, and the presence of these filters will cause SNR loss.

\section{Conclusion}\label{Sec:conclusion}
We have discussed the modeling of the continuous-time AWGN channel in the presence of white (memoryless) phase noise.
We have found that, under mild conditions, the sufficient statistic for this channel is the output of a baud-sampled matched filter,
and that the channel is equivalent to a discrete-time AWGN channel with an SNR penalty.

\section*{Acknowledgment}
L. Barletta was supported by Technische Universit\"{a}t M\"{u}nchen –- Institute for Advanced Study, funded by the German Excellence Initiative.
G. Kramer was supported by
an Alexander von Humboldt Professorship endowed by the
German Federal Ministry of Education and Research.

\bibliographystyle{IEEEtran}
\bibliography{refSuffStat}

\end{document}